\documentclass[%
 reprint,
 amsmath,amssymb,
 aps,
]{revtex4-1}

\usepackage{graphicx}
\usepackage{dcolumn}
\usepackage{bm}
\usepackage{subfigure}
\usepackage{braket,ntheorem}
\usepackage{xcolor}
\usepackage{tabularx}

\newtheorem{theorem}{Theorem}
\newtheorem{lemma}{Lemma}

\newtheorem*{proof}{Proof}

\newcommand{\comment}[1]{{}}
\newcommand{\Tr}{\mathrm{Tr}}

\begin{document}

\preprint{APS/123-QED}

\title{Equivalence checking of quantum circuits by nonlocality}

 \author{Weixiao Sun$^{1}$}
 \author{Zhaohui Wei$^{2,3,}$}\email{weizhaohui@gmail.com}
 \affiliation{$^{1}$Institute for Interdisciplinary Information Sciences, Tsinghua University, Beijing 100084, China\\$^{2}$Yau Mathematical Sciences Center, Tsinghua University, Beijing 100084, People's Republic of China\\$^{3}$Yanqi Lake Beijing Institute of Mathematical Sciences and Applications, 101407, People's Republic of China}

\begin{abstract}
Suppose two quantum circuit chips are located at different places, for which we do not have any prior knowledge, and cannot see the internal structures either. If we want to find out whether they have the same functions or not with certainty, what should we do? In this paper, we show that this realistic problem can be solved completely from the viewpoints of quantum nonlocality. Specifically, we design an elegant protocol that examines underlying quantum nonlocality, where the strongest nonlocality can be observed if and only if two quantum circuits are equivalent to each other. We show that the protocol also works approximately, where the distance between two quantum circuits can be calculated accurately by observed quantum nonlocality in an analytical manner. Furthermore, it turns out that the computational cost of our protocol is independent in the size of compared quantum circuits. Lastly, we also discuss the possibility to generalize the protocol to multipartite cases, i.e., if we do equivalence checking for multiple quantum circuits, we try to solve the problem in one go.
\end{abstract}

\maketitle

\section{Introduction}

In the past several years, the physical realizations of quantum computing have achieved remarkable progresses~\cite{arute2019quantum,zhong2020quantum}. As a result, the following three tasks have become more and more important issues in quantum computing. First, to run a quantum algorithm, which is usually designed in the language of quantum circuit, on a quantum computer, we have to compile it into a series of quantum instructions that can be executed directly on the quantum hardware, and as a whole this is essentially another quantum circuit. Second, when executing quantum instructions on a quantum computer, the hardware configuration has to be respected, which means that the available quantum instructions are actually restricted. If this is not the case, we have to map the quantum circuit at hand into another desirable one. Third, for now the scaling of quantum computing is still small, and quantum computational resources are very precious, therefore it is always nice to make sure that the executed quantum circuit has been optimized. Fourth, quantum computing has been physically implemented on different quantum platforms, then if we run a same quantum algorithm on different platforms, an important problem is to make sure they are essentially the same, where the quantum circuits may look different.

It is not hard to see that a common part in the above four fundamental problems is that we need to transfer a quantum circuit into another, or compare two quantum circuits. Undoubtedly, during these transformations or comparisons, a basic requirement is to find out whether an initial quantum circuit and the compiled, optimized, or compared quantum circuit have exactly the same functions. As a consequence, equivalent checking of quantum circuits is a profound problem in quantum computing and quantum engineering. We stress that sometimes the compared two quantum circuits are located at different places.

In fact, this problem has attracted a lot of attentions, and quite a few approaches have been proposed accordingly. Particularly, in \cite{viamontes2007checking} an approach based on decision diagrams was proposed for equivalence checking of quantum circuits, where the central idea is representing quantum circuits as decision programs, on which the comparisons are performed. In \cite{yamashita2010fast}, a concept called reversible miter was proposed for this problem, which is a generalization of miter circuits utilized in digital electronic circuits, and can be integrated with circuit simplifications and decision programs techniques. Meanwhile, as mentioned above, equivalence checking of quantum circuits have been extensively studied in the optimization of quantum circuits and the verifications of quantum compilers~\cite{amy2014polynomial,nam2018automated,kissinger2019reducing,smith2019quantum,shi2019contract,hietala2021verified}. Very recently, equivalence checking has also been introduced to handle sequential quantum circuits, where a Mealy machine-based framework was proposed~\cite{wang2018equivalence}.

Despite these encouraging approaches for equivalence checking of quantum circuits, however, they share the common feature that internal structures of involved quantum circuits can be seen. If we use the language of software testing, this is essentially a kind of white-box testing. Then like in software testing, black-box testing that the internal structures of quantum circuits cannot be seen should also be an realistic scenario that needs to be considered.

Indeed, as mentioned in future it will be an important problem for us to find out whether two \emph{separated} manufactured quantum circuits chips that the insides cannot be seen have the same functions \emph{with certainty}. Trying to solve this problem is the main target of the current paper. We stress that in our setting we do not have any prior knowledge on quantum circuits to be compared, and this is essentially different from the topic of unitary operation discrimination~\cite{acin2001statistical,d2001using,duan2007entanglement}, where every unitary operation is picked up from a small set known beforehand.

In this paper, based on the key role played by quantum nonlocality, we design an elegant approach that can achieve black-box equivalence checking of quantum circuits with certainty. Clearly, no similar approach exists for the classical counterpart of this problem. Particularly, we provide a complete mathematical characterization for our approach. First, we prove that in our protocol, the observed quantum nonlocality is the strongest if and only if the two involved quantum circuits have exactly the same functions. Second, we show that the protocol also works well in an approximate sense, i.e., for a given strength of observed quantum nonlocality, we provide \emph{analytical} lower and upper bounds for the distance between the two quantum circuits. By providing numerical evidences, we verify the correctness of these bounds. Third, by looking into the structure of the gap between the above two bounds, we proposed a modified protocol such that the gap disappears, which means that based on the observed nonlocality we can completely pin down the distance between the compared quantum circuits generally. Fourth, we analyze the computational cost of the modified protocol, and show that it is independent of the size of compared quantum circuits. That is, for a given precision we need only a \emph{constant} cost to check equivalence of large quantum circuits. Lastly, we discuss the possibility to generalize our protocol to the case of multiple quantum circuits, where we want to determine whether three or even more quantum circuits are equivalent to each other in one go. We argue that at least when the number of quantum circuits is odd, this is impossible. We believe that our results demonstrate a new possibility to apply quantum nonlocality to important problems in future quantum engineering.

\section{The exact equivalence checking of two quantum circuits}

Suppose two $n$-qubit quantum circuits $C_1$ and $C_2$ are held by two separated players, Alice and Bob, respectively. Since the Hadamard gate and the Toffoli gate form a universal gate set for quantum computation~\cite{shi2003both}, without loss of generality we suppose that the matrix representations of $C_1$ and $C_2$ are \emph{real}, denoted $U_1$ and $U_2$. Then our task is to determine whether $U_1$ is equivalent to $U_2$ up to a global phase (since they are real, a global phase can only be $\pm 1$). Let us first consider the smallest case where $C_1$ and $C_2$ are single-qubit quantum circuits.

Before introducing our main idea, let us recall some facts on quantum nonlocality and Bell experiments. Suppose Alice and Bob share a lot of EPR paris, i.e., $\ket{\text{EPR}}=\frac{1}{\sqrt{2}}(\ket{00}+\ket{11})$. On each EPR pair, they repeat the following procedure. Both of them perform random local measurements on their qubits respectively, where Alice measures observables $A_0=\sigma_X$ and $A_1=\sigma_Z$, and Bob measures observables $B_0=(\sigma_X+\sigma_Z)/\sqrt{2}$ and $B_1=(\sigma_X-\sigma_Z)/\sqrt{2}$. Here $\sigma_X$ and $\sigma_Z$ are Pauli matrices. Then they calculate all the probability distribution $p(ab|xy)$, i.e., the probability that Alice and Bob obtain outcomes $a$ on $A_x$ and $b$ on $B_y$ respectively, where $a,b\in\{-1,1\}$ and $x,y\in\{0,1\}$. Let $\langle A_xB_y\rangle=\sum_{a,b}ab\cdot p(ab|xy)$, and
\begin{equation}
I_{\text{CHSH}}=\langle A_0B_0\rangle+\langle A_1B_0\rangle+\langle A_0B_1\rangle-\langle A_1B_1\rangle,
\end{equation}
then it holds that $I_{\text{CHSH}}=2\sqrt{2}$. As a comparison, if $p(ab|xy)$ is produced by a classical system, the corresponding value will not be larger than $2$, and this is the famous Clauser-Horne-Shimony-Holt (CHSH) inequality~\cite{clauser1969proposed}. A well-known fact is that the above violation to the CHSH inequality achieved by EPR pairs is optimal~\cite{clauser1969proposed}, which is the foundation of many quantum information processing tasks~\cite{popescu1992states,mayers1998quantum,ekert1991quantum}.

We now change the above Bell experiment a little bit by adding one more step. Before measuring each EPR pair, Alice and Bob input the qubit they hold into $C_1$ and $C_2$ respectively, then the overall output will be $\ket{\psi}=\frac{1}{\sqrt{2}}(U_1\ket{0}\otimes U_2\ket{0}+U_1\ket{1}\otimes U_2\ket{1})$, on which they perform the \emph{same} sets of local measurements as above. Here we stress that it is crucial to use the same sets of local measurements. We now analyze the new value of $I_{\text{CHSH}}$, denoted $I'_{\text{CHSH}}$.

We first consider the case that $U_1=U_2$. Recall that they are real unitary matrices, then it can be verified that $\ket{\psi}=\ket{\text{EPR}}$, which means $I'_{\text{CHSH}}=2\sqrt{2}$. That is to say, if $C_1$ and $C_2$ are the same, the above experiment will still result in a maximal violation. In this situation, it is natural to ask, is the converse correct? That is, does $I'_{\text{CHSH}}=2\sqrt{2}$ always imply that $U_1=U_2$? If this is correct, then we can perfectly determine whether $C_1$ and $C_2$ are equivalent by performing the above modified Bell experiment.

Actually, this is indeed the case. It has been known that if $I'_{\text{CHSH}}=2\sqrt{2}$, the following conditions are satisfied~\cite{popescu1992states}.
\begin{equation}
\frac{A_0\pm A_1}{\sqrt{2}}\ket{\psi}=B_{0/1}\ket{\psi}.
\end{equation}
By straightforward calculations, it can be verified that this indicates that $\ket{\psi}=\ket{\text{EPR}}$ up to a global phase. On the other hand, if $U_1\neq\pm U_2$, it can be checked that $\ket{\psi}\neq\pm \ket{\text{EPR}}$, which means that if $I'_{\text{CHSH}}=2\sqrt{2}$, we must have $U_1=U_2$.

We now move to the general case, where the common size of $C_1$ and $C_2$ is $n$ qubits. Let $d=2^n$. Inspired by the single-qubit case, Alice and Bob hope they can use a similar protocol to find out whether $C_1$ and $C_2$ are equivalent. That is, they hope that the following plan could be realized. Again, they first prepare and share many copies of the maximally entangled state
\begin{equation}
\ket{\Phi_d}=\frac{1}{\sqrt{d}}\sum_{i=0}^{d-1}\ket{ii},
\end{equation}
and choose a certain Bell inequality such that $\ket{\Phi_d}$ violates it maximally, where they record the local measurements that achieve the maximal violation. Then for each copy of $\ket{\Phi_d}$, Alice and Bob input their own subsystems into the corresponding quantum circuits they hold respectively. On the output state, which is now $(U_1\otimes U_2)\ket{\Phi_d}$, they perform the same local measurements as recorded above. By repeating the experiments, they collect the measurement outcome statistics data $p(ab|xy)$, where $x,y\in\{1,2,...,m\}$ and $a,b\in\{0,1,...,d-1\}$ are the labels for the local measurements and the corresponding outcomes. Then they examine the measurement outcome statistics data with the above chosen Bell inequality, and hope that $(U_1\otimes U_2)\ket{\Phi_d}$ violates the Bell inequality maximally if and only if $U_1=U_2$ up to a global phase.

Clearly, if the above Bell equality exists, like in the qubit case, Alice and Bob can determine whether $C_1$ and $C_2$ are equivalent perfectly according to the violation. Again, a key question is, can we find such a Bell inequality when $n>1$? Interestingly, it turns out that the answer is positive.

According to our plan, such a desirable Bell inequality should be violated maximally by maximally entangled states. However, it has been well-known that entanglement is a different resource from quantum nonlocality, and on many Bell inequalities it is not maximally entangled states that achieve the maximal violations, say the Collins-Gisin-Linden-Masser-Popescu (CGLMP) inequalities~\cite{collins2002bell}. In the meantime, quantum nonlocality can be observed directly by quantum experiments, while entanglement cannot, thus we often choose to characterize unknown entanglement by looking into the underlying quantum nonlocality. Therefore, when doing this, we hope that quantum nonlocality we observed and the underlying entanglement are as consistent as possible, which implies that the above desirable Bell inequalities will be nice choices. Fortunately, in \cite{salavrakos2017bell} such a class of beautiful Bell inequalities have been proposed, which were deliberately designed to be violated maximally by $\ket{\Phi_d}$.

Specifically, to perform the measurement labelled by $x$, Alice measures an observable with eigenvectors $\ket{a}_x$ ($a=0,1,...,d-1$, and $x=1,2,...,m$), and
\begin{equation}\label{AliceEigenvector}
\ket{a}_x=\frac{1}{\sqrt{d}}\sum_{k=0}^{d-1}\text{exp}\left[\frac{2\pi \mathbf i}{d}k(a-\alpha_x)\right]\ket{k},
\end{equation}
where ${\mathbf i}=\sqrt{-1}$ is the imaginary number, and $\alpha_x=(x-1/2)/m$. Similarly, to perform the measurement labelled by $y$, Bob measures an observable with eigenvectors $\ket{b}_y$ ($b=0,1,...,d-1$, and $y=1,2,...,m$), and
\begin{equation}\label{BobEigenvector}
\ket{b}_y=\frac{1}{\sqrt{d}}\sum_{k=0}^{d-1}\text{exp}\left[-\frac{2\pi \mathbf i}{d}k(b-\beta_y)\right]\ket{k},
\end{equation}
where $\beta_y=y/m$. On an arbitrary quantum state $\ket{\phi}$, the Bell expression is essentially equivalent to
\begin{equation}\label{BellExpression}
I_{d,m}(\ket{\phi})=\sum_{i=1}^m\sum_{l=1}^{d-1}\bra{\phi}(A_i^l\otimes \bar{B}_i^l)\ket{\phi},
\end{equation}
where $A_i^l=\sum_{a=0}^{d-1}\omega^{al}\ket{a}_{ii}\bra{a}$, $\bar{B}_i^l=(A_i^l)^*$, and $\omega=\text{exp}(2\pi \mathbf i/d)$. Note that $A_i^l$ and $\bar{B}_i^l$ are unitary matrices.

In \cite{salavrakos2017bell}, it was proved that the Tsirelson bound of $I_{d,m}$ is $m(d-1)$, which is achieved exactly by $\ket{\Phi_d}$ and strictly larger than the classical bound. Indeed, a property of $\ket{\Phi_d}$ is that for any $d\times d$ matrices $M$ and $N$, it holds that $(M\otimes N)\ket{\Phi_d}=(I\otimes NM^T\ket{\Phi_d})$. Since $\bar{B}_i^l=(A_i^l)^*$ for any $i$ and $l$, we have that $\bra{\Phi_d}(A_i^l\otimes \bar{B}_i^l)\ket{\Phi_d}=\bra{\Phi_d}(I\otimes I)\ket{\Phi_d}=1$, implying that $I_{d,m}=m(d-1)$ on this state.

Let us go back to our task. We first notice that if $C_1$ and $C_2$ are the same, i.e., $U_1=U_2=U$, $(U_1\otimes U_2)\ket{\Phi_d}$ always achieves the Tsirelson bound of $I_{d,m}$. In fact, for any $i$ and $l$ it holds that
\begin{align*}
	& \bra{\Phi_d}(U^T\otimes U^T)(A_i^l\otimes \bar{B}_i^l)(U\otimes U)\ket{\Phi_d}\\
       =&\bra{\Phi_d}(I\otimes U^T\bar{B}_i^lUU^T(A_i^l)^TU)\ket{\Phi_d} \\
	   =& \bra{\Phi_d}(I\otimes I)\ket{\Phi_d}\\
       =& 1.
\end{align*}
Hence, the new value of $I_{d,m}$ is still $m(d-1)$. In this situation, similar to the case of single-qubit quantum circuits, we need to consider whether the converse is correct or not. Or, can we have $U_1\neq U_2$ but $I_{d,m}((U_1\otimes U_2)\ket{\Phi_d})=m(d-1)$? We now show that this is impossible.
\begin{theorem}
$I_{d,m}((U_1\otimes U_2)\ket{\Phi_d})=m(d-1)$ if and only if $U_1=U_2$ up to a global phase.
\end{theorem}
\begin{proof}
We only need to prove that $I_{d,m}((U_1\otimes U_2)\ket{\Phi_d})=m(d-1)$ implies $U_1=U_2$. According to the definition of $I_{d,m}$, we know that if $I_{d,m}((U_1\otimes U_2)\ket{\Phi_d})=m(d-1)$, each term in the summation of Eq.\eqref{BellExpression} will be $1$. Therefore, for any $i\in\{1,2,...,m\}$ it holds that (let $l=1$)
\begin{align*}
	&\bra{\Phi_d}(U_1^T\otimes U_2^T)(A_i^1\otimes \bar{B}_i^1)(U_1\otimes U_2)\ket{\Phi_d} \\
      = & \bra{\Phi_d}(I\otimes U_2^T\bar{B}_i^1U_2U_1^T(A_i^1)^TU_1)\ket{\Phi_d}\\
       =&\frac{1}{d}\Tr(U_2^T\bar{B}_i^1U_2U_1^T(A_i^1)^TU_1)\\
	   =&\frac{1}{d}\Tr(U_1U_2^T\bar{B}_i^1U_2U_1^T(A_i^1)^T)\\
       =& 1,
\end{align*}
where we have utilized the fact that for any $d\times d$ matrices $M$ and $N$, it holds that $\bra{\Phi_d}(I\otimes M)\ket{\Phi_d}=\Tr(M)/d$ and $\Tr(MN)=\Tr(NM)$. Hence, we obtain that $\Tr(U_1U_2^T\bar{B}_i^1U_2U_1^T(A_i^1)^T)=d$.

Meanwhile, note that $U_1U_2^T\bar{B}_i^1U_2U_1^T(A_i^1)^T$ is a $d\times d$ unitary matrix, thus we have that $U_1U_2^T\bar{B}_i^1U_2U_1^T(A_i^1)^T=I$.
For simplicity, let $S_1=U_2U_1^T$ and $S_2=\bar{B}_i^1$. Then this means $S_1^{\dag}S_2S_1S_2^{\dag}=I$, which is also $S_2S_1=S_1S_2$, where we have utilized the fact that both $S_1$ and $S_2$ are unitary matrices. Since $S_1$ and $S_2$ are also normal matrices, this shows that they can be simultaneously diagonalizable.

Similarly, let $j\neq i\in\{1,2,...,m\}$ and $S_3=\bar{B}_j^1$, then $S_1$ and $S_3$ can also be simultaneously diagonalizable. Recall the definition of $\bar{B}_j^1$, whose eigenvectors are given by the conjugate of Eq.\eqref{AliceEigenvector}, then we have that $S_1$ can be diagonalized in the following two different ways,
\begin{equation}
S_1=\sum_{a=0}^{d-1}g_a(\ket{a}_{ii}\bra{a})^*=\sum_{a=0}^{d-1}h_a(\ket{a}_{jj}\bra{a})^*,
\end{equation}
where for any $a$, $g_a$ and $h_a$ are unit complex number. Then
\begin{equation}
_{i}^{*}\bra{0}S_1\ket{0}_i^{*}=g_0=\sum_{a=0}^{d-1}h_a\cdot|_{i}^{*}\bra{0}a\rangle_j^{*}|^2.
\end{equation}

At the same time, for any $a\in\{0,1,...,d-1\}$ it can be verified that $0<|_{i}^{*}\bra{0}a\rangle_j^{*}|^2<1$. Combining this with the fact that $\sum_{a=0}^{d-1}|_{i}^{*}\bra{0}a\rangle_j^{*}|^2=1$, we obtain that there exists a $\gamma\in[0,2\pi)$ such that $g_0=h_0=...=h_{d-1}=e^{\mathbf i\gamma}$, which implies that $S_1=e^{\mathbf i\gamma}\cdot I$. According to the definition of $S_1$, we now have that $U_1=U_2$ up to a global phase, which completes the proof.
\end{proof}

The theorem shows the correctness of our plan, and we can indeed determine whether $U_1$ and $U_2$ have the same function by examining the underlying quantum nonlocality of $(U_1\otimes U_2)\ket{\Phi_d}$.

\section{The approximate case}

Since equivalent checking is an important issue in engineering applications, we need to address the situation that quantum circuits are realized approximately. For example, unitary operations $U_1$ and $U_2$ correspond to two different quantum circuits for a same quantum algorithm, hence they are supposed to be the same. However, due to certain mistakes one of the quantum circuits contains some more quantum gates, which implies that $U_1\neq U_2$. Here for simplicity we suppose the error in realizing quantum circuits are unitary errors. Note that this form of error covers the case that the preparation of $\ket{\Phi_d}$ is also affected by local unitary errors. Our numerical simulations show that more general form of weak errors that are expressed as quantum operations can also be handled, though it is hard to provide analytical discussions like in the unitary case below.

Since $U_1\neq U_2$, if we do the Bell experiment introduced previously using $U_1$ and $U_2$, the Bell expression value $I_{d,m}((U_1\otimes U_2)\ket{\Phi_d})$ will be not exactly $m(d-1)$. In this situation, an interesting question is, can we draw any nontrivial conclusions on $D(U_1,U_2)$, the distance between $U_1$ and $U_2$ based on the value of $I_{d,m}((U_1\otimes U_2)\ket{\Phi_d})$? We now show that this is indeed the case, and furthermore, $D(U_1,U_2)$ can be lower and upper bounded {analytically}.

In this paper, we choose the definition for $D(U_1,U_2)$ given by \cite{montanaro2016survey}, which is
\begin{equation}
D(U_1,U_2)=\sqrt{1-\left|\frac{1}{d}\Tr(U_1^TU_2)\right|^2}.
\end{equation}
Meanwhile, we need to use the following key fact (see Appendix A for its proof).
\begin{lemma}\label{lem:eigenvalues}
Suppose $\ket{\psi}$ is a $d\times d$ quantum state orthogonal to $\ket{\Phi_d}$. Then
\begin{equation}
-m\leq I_{d,m}(\ket{\psi})\leq m(d-2).
\end{equation}
\end{lemma}
Having this fact, we are ready to give the second main result of the current paper.
\begin{theorem}\label{thm:approximate}
Suppose $V=I_{d,m}((U_1\otimes U_2)\ket{\Phi_d})$, then we have that
\begin{equation}
\sqrt{1-\frac{V+m}{md}}\leq D(U_1,U_2)\leq \sqrt{1-\frac{V-m(d-2)}{m}}.
\end{equation}
\label{thm:approx}
\end{theorem}
\begin{proof}
Let $\ket{\alpha}=(U_1\otimes U_2)\ket{\Phi_d}=(I\otimes U_2U_1^T)\ket{\Phi_d}$. Suppose an orthogonal decomposition of $U_2U_1^T$ is $U_2U_1^T=\sum_{j=0}^{d-1}e^{\mathbf i\theta_j}\ket{\lambda_j}\bra{\lambda_j}$, where $\theta_j\in[0,2\pi)$. Note that we also have $\ket{\Phi_d}=\sum_{j=0}^{d-1}\ket{\lambda_j}\ket{\lambda_j}^*/\sqrt{d}$. Therefore, we have that
\begin{equation}
\ket{\alpha}=\sum_{j=0}^{d-1}e^{\mathbf i\theta_j}\ket{\lambda_j}\ket{\lambda_j}^*/\sqrt{d}.
\end{equation}
Let $\ket{\alpha}=c_1\ket{\Phi_d}+c_2\ket{\Phi^\bot}$, where $c_1$ and $c_2$ are complex numbers, $|c_1|^2+|c_2|^2=1$, and $\langle \Phi^\bot\ket{\Phi_d}=0$. Then it can be seen that
\begin{equation}
c_1=\langle \Phi_d|\alpha\rangle = \sum_{j=0}^{d-1}\frac{e^{\mathbf i\theta_j}}{d}=\frac{\Tr(U_2U_1^T)}{d},
\end{equation}
which means that $D(U_1,U_2)^2=1-|c_1|^2$.

For convenience, let $B=\sum_{i=1}^m\sum_{l=1}^{d-1}(A_i^l\otimes \bar{B}_i^l)$. Then it holds that
\begin{align*}
	   V = & \bra{\alpha}B\ket{\alpha}=(c_1^*\bra{\Phi_d}+c_2^*\bra{\Phi^\bot})B(c_1\ket{\Phi_d}+c_2\ket{\Phi^\bot})\\
       =&|c_1|^2\bra{\Phi_d}B\ket{\Phi_d}+|c_2|^2\bra{\Phi^\bot}B\ket{\Phi^\bot}\\
	   =&|c_1|^2\cdot m(d-1) +(1-|c_1|^2)\bra{\Phi^\bot}B\ket{\Phi^\bot}.
\end{align*}
According to Lemma \ref{lem:eigenvalues}, we have that $-m\leq \bra{\Phi^\bot}B\ket{\Phi^\bot}\leq m(d-2)$, which means that
\begin{equation}
\sqrt{\frac{V-m(d-2)}{m}}\leq |c_1|\leq\sqrt{\frac{V+m}{md}}.
\end{equation}
Combining this with the fact that $D(U_1,U_2)^2=1-|c_1|^2$, we complete the proof.
\end{proof}

Note that when $V=m(d-1)$, both the lower and the upper bounds are exactly $1$, implying that both of them are tight in this case. When $V$ does not achieve $m(d-1)$, the lower bound for $D(U_1,U_2)$ reveals the minimum distance between $U_1$ and $U_2$, thus in some sense it is more informative than the upper bound.

To examine the performance of the above analytical bounds, we test them with numerical simulations. For this, we generate many random instances for $U_1$ and $U_2$, then for each pair of $U_1$ and $U_2$ we compute the corresponding exact values of $D(U_1,U_2)$, which are next compared with the lower and upper bounds for $D(U_1,U_2)$ given by Theorem \ref{thm:approximate}. The results are listed in Fig.\ref{fig:result_approx}, where it can be seen that the lower bound is quite tight in many instances.
\begin{figure}[!ht]
	\centering
	\includegraphics[width=0.5\textwidth]{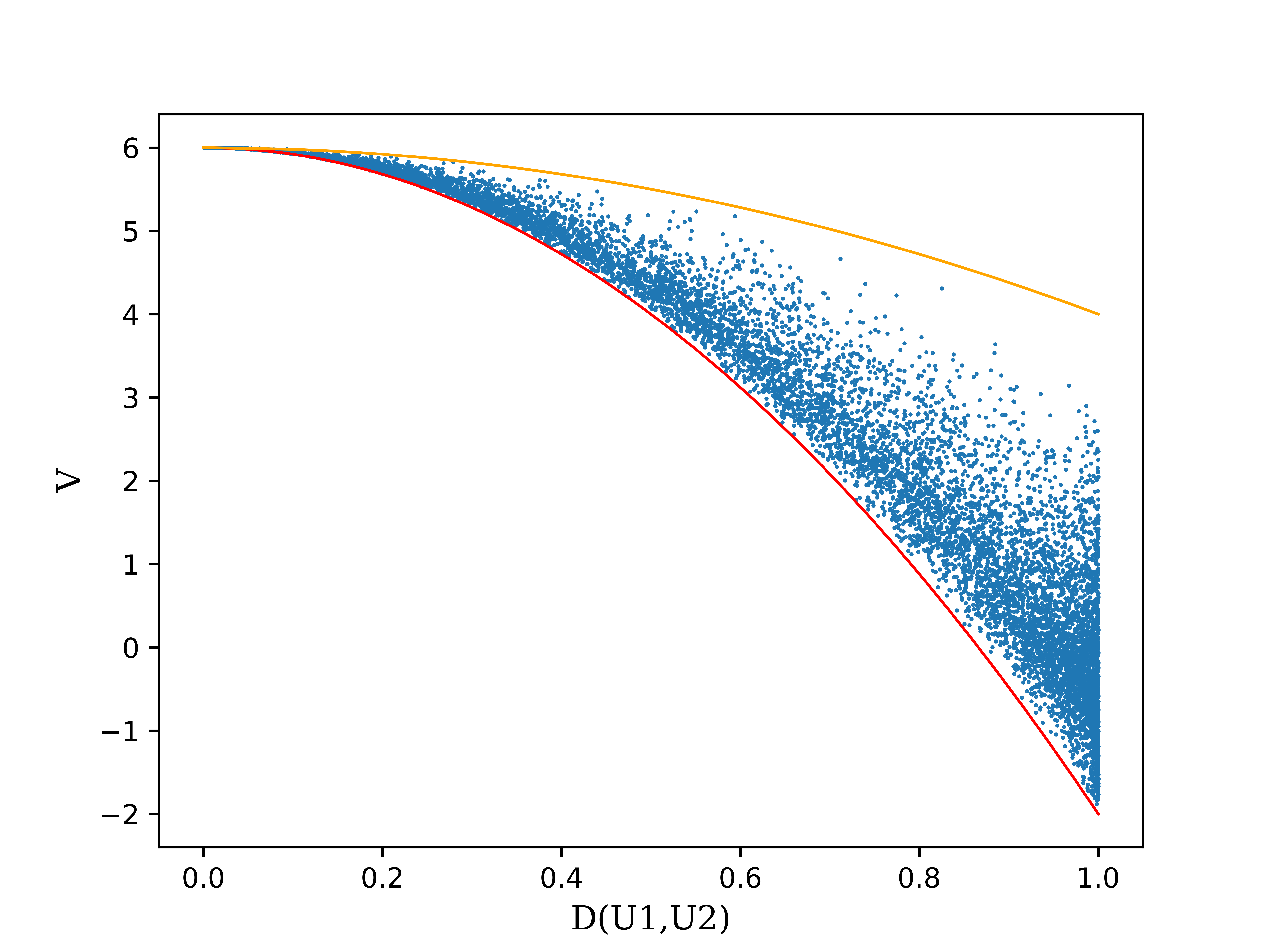}
	\caption{The values of $D(U_1,U_2)$ and $V$, where $d=4$ and $m=2$. Here each blue point represents a pair of $U_1$
     and $U_2$ that is randomly generated, on which the exact values of $D(U_1,U_2)$ and $V$ are given. The red and the orange solid lines are respectively the lower and the upper bounds provided by Theorem \ref{thm:approx}.}
	\label{fig:result_approx}
\end{figure}

\section{Direct determination of the distance $D(U_1,U_2)$}

In Fig.\ref{fig:result_approx}, it can be observed that in most cases the upper bound for $D(U_1,U_2)$ given by Theorem \ref{thm:approximate} is quite loose compared with the lower bound. From the proof for Theorem \ref{thm:approximate}, it can be seen that the reason is that the bound $\bra{\Phi^\bot}B\ket{\Phi^\bot}\leq m(d-2)$ we have utilized is far from tight in most cases. If we could somehow improve the upper bound for $\bra{\Phi^\bot}B\ket{\Phi^\bot}$, our estimation for $D(U_1,U_2)$ will be more accurate accordingly.

To understand the behave of $\bra{\phi}B\ket{\phi}$, we studied its value for a uniformly random pure state $\ket{\phi}$. It turns out that $\bra{\phi}B\ket{\phi}$ is very small with probability close to $1$. Particularly, we have the following fact, and its proof can be seen in Appendix B.
\begin{lemma}
    Given $0<\delta<1$. Suppose $|\psi\rangle$ is a $d\times d$ quantum state, which as a unit vector is chosen uniformly at random on the $d^2$-dimensional real unit sphere. Then with the probability of no less than $1-\delta$ it holds that
    \begin{equation}
        I_{d,m}(|\psi\rangle)\le m\sqrt{\frac{4}{3d\delta}}.
    \end{equation}
\end{lemma}

Though for a random pair $U_1$ and $U_2$, it is possible that the distribution of $\ket{\Phi^\bot}$ is not uniformly random, the above lemma still helps us to understand why the estimation $\bra{\Phi^\bot}B\ket{\Phi^\bot}\leq m(d-2)$ is quite loose overall. Inspired by this, we now adjust the structure of our protocol, and the purpose is to make sure that the new value of $\bra{\Phi^\bot}B\ket{\Phi^\bot}$ is low.

Again $U_1$ and $U_2$ are the two $n$-qubit circuits that we want to compare. Now we construct a $2n$-qubit circuit as shown in Fig.\ref{fig:determination_circuit}, and denote it as $U_1'$, where $U_1$ is a part of $U_1'$. And $U_2'$ is constructed similarly. Then we apply our protocol to compare the new quantum circuits $U_1'$ and $U_2'$, whose size is now larger.

We now prove that this adjustment will pin down the new value of $\bra{\Phi^\bot}B\ket{\Phi^\bot}$ to be $-m$, which is actually the smallest possible. As a result, the upper bound for $D(U_1',U_2')$ given by Theorem \ref{thm:approximate} now matches the lower bound completely. That is to say, from the value of Bell expression $I_{d,m}((U_1'\otimes U_2')|\Psi_d\rangle)$, $D(U_1,U_2)=D(U_1',U_2')$ can be determined directly, where $d=2^{2n}$.

\begin{figure}[!ht]
	\centering
	\includegraphics[width=0.3\textwidth]{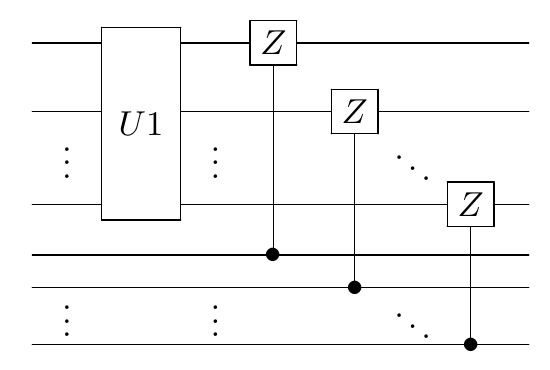}
	\caption{The $2n$-qubit circuit $U_1'$, which contains the $n$-qubit circuit $U_1$. }
	\label{fig:determination_circuit}
\end{figure}

\begin{theorem}
    \label{thm:directly}
    Suppose $V=I_{d,m}((U_1'\otimes U_2')\ket{\Phi_d})$ where $d=2^{2n}$, then we have that
    \begin{equation}
    D(U_1,U_2)=D(U_1',U_2')=\sqrt{1-\frac{V+m}{md}}.
    \end{equation}
\end{theorem}

\begin{proof}
    Denote the operation of all the control-$Z$ gates combined in Fig.\ref{fig:determination_circuit} by $U_{Z}$ (as a unitary matrix on $2n$ qubits). That is $U_1'=U_{Z}(U_1\otimes I), U_2'=U_{Z}(U_2\otimes I)$. Then
    \begin{align*}
        D(U_1', U_2')&=\sqrt{1-\left|\frac{1}{2^{2n}}\Tr(U_1'^TU_2')\right|}\\
        &=\sqrt{1-\left|\frac{1}{2^{2n}}\Tr((U_1^T\otimes I)U_Z^TU_{Z}(U_2\otimes I))\right|}\\
        &=\sqrt{1-\left|\frac{1}{2^{2n}}\Tr((U_1^T\otimes I)(U_2\otimes I))\right|}\\
        &=\sqrt{1-\left|\frac{1}{2^{2n}}\Tr(U_1^TU_2\otimes I)\right|}\\
        &=\sqrt{1-\left|\frac{1}{2^{n}}\Tr(U_1^TU_2)\right|}\\
        &=D(U_1,U_2).
    \end{align*}

    In the proof for Lemma \ref{lem:eigenvalues} (see Appendix A), we have already known that if we let $(U_1'\otimes U_2')\ket{\Phi_d}=\sum_{k=0}^{d-1}\sum_{j=0}^{d-1}\gamma_{kj}|k\rangle|j\rangle$, it holds that
    \begin{align*}
        I_{d,m}((U_1'\otimes U_2')\ket{\Phi_d})=&m\sum_{r=0}^{d-1}\left(\left|\sum_{k=0}^{d-r-1}\gamma_{k(k+r)}\right|^2\right. \\
        &+\left.\left|\sum_{k=d-r}^{d-1}\gamma_{k(k+r-d)}\right|^2\right)-m.
    \end{align*}
    Now let us notice the following properties of $\gamma_{kj}$. Let $k=\overline{a_1a_2\dots a_nb_1b_2\dots b_n}$ and $j=\overline{c_1c_2\dots c_nd_1d_2\dots d_n}$ be binary representations of $k$ and $j$, where $a_i,b_i,c_i,d_i\in\{0,1\}$ for $1\leq i\leq n$. Then based on the construction of $U_1'$ and $U_2'$ given by Fig.\ref{fig:determination_circuit}, it can be verified that
    \begin{enumerate}
        \item If $\overline{a_1a_2\dots a_n}\neq\overline{c_1c_2\dots c_n}$, then $\gamma_{kj}=0$;
        \item If $\overline{a_1a_2\dots a_n}=\overline{c_1c_2\dots c_n}$ and $b_i\neq d_i$, let $k'=\overline{a_1a_2\dots a_{i-1}(1-a_i)a_{i+1}\dots a_nb_1b_2\dots b_n},j'=\overline{c_1c_2\dots c_{i-1}(1-c_i)c_{i+1}\dots c_nd_1d_2\dots d_n}$, then $\gamma_{k'j'}=-\gamma_{kj}$,
            where we have utilized the facts that only one of $a_i$ and $1-a_i$ can trigger the $Z$ operators on the positions $b_i$ and $d_i$ and that $b_i\neq d_i$.
    \end{enumerate}
    By using the properties repeatedly, one can prove that $\sum_{k=0}^{d-r-1}\gamma_{k(k+r)}=\sum_{k=d-r}^{d-1}\gamma_{k(k+r-d)}=0$ when $r\neq0$. Thus we have that
    \begin{align*}
        I_{d,m}((U_1'\otimes U_2')\ket{\Phi_d})&=m\left|\sum_{k=0}^{d-1}\gamma_{kk}\right|^2-m\\
        &=md\left|\bra{\Phi_d}(U_1'\otimes U_2')\ket{\Phi_d}\right|^2-m\\
        &=md(1-D(U_1',U_2')^2)-m.
    \end{align*}
    That is, $D(U_1',U_2')=\sqrt{1-\frac{V+m}{md}}$.
\end{proof}

Therefore, to determine the distance between two $n$-qubit quantum circuits, we can embed them into two larger $2n$-qubit quantum circuits and then apply our original protocol on the latter. Though the cost is a little bit higher, the estimation for the distance can be much more accurate. We also perform numerical simulations to verify our modified protocol, where again random $U_1$ and $U_2$ are sampled. The results are listed in Fig.\ref{fig:result_cost}.

\begin{figure*}[!ht]
	\centering
	\includegraphics[width=0.85\textwidth]{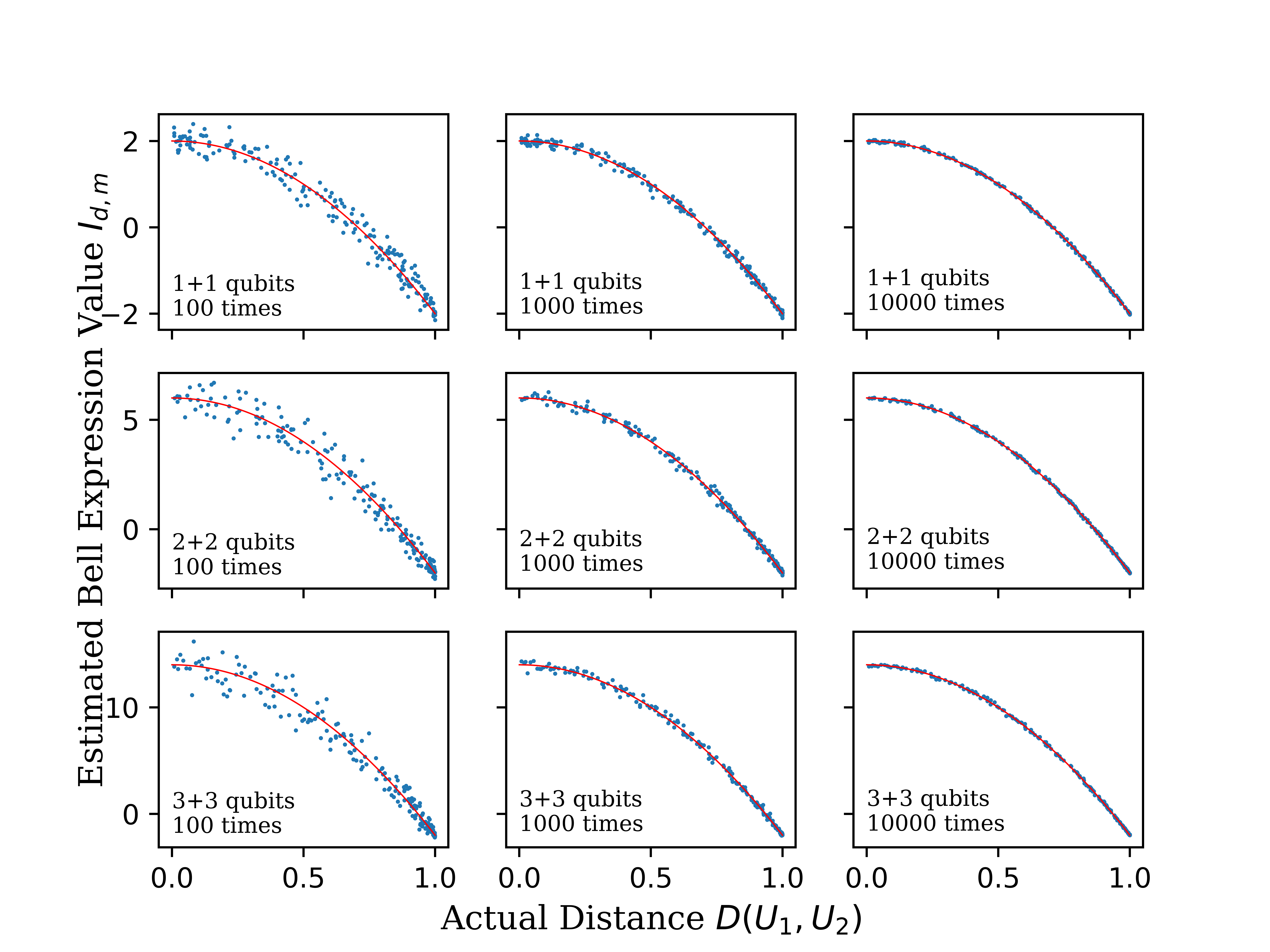}
	\caption{Determining the distance $D(U_1,U_2)$ based on Theorem \ref{thm:directly}, where $U_1$ and $U_2$ are randomly sampled, and $m=2$. Every blue point corresponds to picking up a specific pair of $U_1$ and $U_2$, and then running the Mento Carlo process described in section \ref{sec:cost}, where $s=100,1000,10000$ respectively. We sample $U_1$ and $U_2$ of 1, 2 or 3 qubits, which means the experiments are performed on two 2, 4 or 6-qubit circuits. The red solid line is given by the Theorem \ref{thm:directly}.}
	\label{fig:result_cost}
\end{figure*}

\section{The analysis of computational cost}
\label{sec:cost}

Now let us analyze the computational cost of our modified protocol, that is, the number of times that we have to run the unknown circuits in order to give a good estimation of the distance $D(U_1,U_2)$ based on Theorem \ref{thm:directly}. For convenience, we reformulate the Bell expression as below, and the corresponding details can be found in~\cite{salavrakos2017bell}.
\begin{align}
    I_{d,m}&=dmI_{d,m}'-m\label{normalization}\\
    I_{d,m}'&=\frac{1}{m}\sum_{k=0}^{d-1}\sum_{i=1}^m\alpha_k[P(A_i=B_i+k)+P(B_i=A_{i+1}+k)]
\end{align}
where $\alpha_k=\frac{1}{2d}\tan(\frac{\pi}{2m})\cot(\frac{\pi}{d}(k+\frac{1}{2m}))$ and $A_{m+1}=A_1+1$. For simplicity, in this section $I_{d,m}$ and $I_{d,m}'$ are short for $I_{d,m}((U_1'\otimes U_2')|\Psi_d\rangle)$ and $I_{d,m}'((U_1'\otimes U_2')|\Psi_d\rangle)$ respectively. Since $-m\le I_{d,m}\le m(d-1)$, we have $0\le I_{d,m}'\le1$. Meanwhile, Theorem \ref{thm:directly} implies that $D(U_1,U_2)=\sqrt{1-I_{d,m}'}$.

Now we consider the estimation of $I_{d,m}'$, where $d=2^{2n}$. First Alice and Bob apply circuits on their own subsystems of the maximally entangled state to get $(U_1'\otimes U_2')|\Psi_d\rangle$. Then choose $r\in\{0,1\}$ and $i\in\{1,2,\dots,m\}$ equiprobably. If $r=0$, Alice and Bob perform measurements $A_i$ and $B_i$ respectively and obtain the outcomes $a$ and $b$, then they return $2\alpha_{a-b \text{ mod } d}$. If $r=1$, Alice and Bob perform measurements $A_{i+1}$ and $B_i$ and obtain the outcomes $a$ and $b$, then they return $2\alpha_{b-a\text{ mod }d}$. They repeat the above process $s$ times. Denote the return values by $X_j$, $j=1,2,\dots,s$. Then it turns out that $X\equiv\frac{1}{s}\sum_{j=1}^sX_j$ is an estimation of $I_{d,m}'$.

Indeed, note that $\mathbb{E}(X_j)=I_{d,m}'$, which means $\mathbb{E}(X)=I_{d,m}'$. Furthermore, since $|\alpha_k|\le1$, by Hoeffding's inequality, if $s>8\log(1/\delta)/\epsilon^2$, we have that
\begin{equation}
    P(|X-I_{d,m}'|\ge \epsilon)\le\delta.
\end{equation}
That is to say, in order to estimate the value of $I_{d,m}'$ within additive error $\epsilon$, the cost of our protocol is $O(\log(1/\delta)/\epsilon^2)$, which is completely independent of the dimension. Then according to Theorem \ref{thm:directly}, if we want to estimate $D=D(U_1,U_2)$ within additive error $\epsilon$, then the cost of our protocol will be $O(\log(1/\delta)/D^2\epsilon^2)$ if $D>\epsilon$, or $O(\log(1/\delta)/\epsilon^4)$ if $0\le D\le\epsilon$.

As a comparison, we can consider an alternative approach to verify whether $U_1$ and $U_2$ are the same, which performs quantum process tomography (QPT) for $U_1$ and $U_2$ separately and then compare the two outputs. The standard QPT technique needs to estimate roughly $O(d^4)$ quantities. Recently, QPT protocols has been customized to characterize unitary operations~\cite{reich2013minimum,baldwin2014quantum}, which reduced the cost to $O(d^2)$. The cost of our protocol is much less than QPT and gets rid of the exponential growth with the number of qubits increasing, which means our protocol is practical in the era of large-scale quantum computation.

Lastly, we would like to stress that the measurements involved in our protocol can be physically implemented by a serial of single-qubit measurements. In fact, it is not hard to verify that the observable eigenvectors given in Eqs.\eqref{AliceEigenvector} and \eqref{BobEigenvector} can always be decomposed as tensor products of single-qubit pure states as below.
\begin{widetext}
\begin{align*}
    |a\rangle_x&=\frac{1}{\sqrt{d}}\sum_{k=0}^{d-1}\exp\left[\frac{2\pi i}{d}k(a-\alpha_x)\right]|k\rangle
    =\bigotimes_{j=1}^n\left(|0\rangle+\exp\left[\frac{2\pi i}{d}2^{j-1}(a-\alpha_x)\right]|1\rangle\right)/\sqrt{2},\\
    |b\rangle_y&=\frac{1}{\sqrt{d}}\sum_{k=0}^{d-1}\exp\left[-\frac{2\pi i}{d}k(b-\beta_y)\right]|k\rangle
    =\bigotimes_{j=1}^n\left(|0\rangle+\exp\left[-\frac{2\pi i}{d}2^{j-1}(b-\beta_y)\right]|1\rangle\right)/\sqrt{2}.
\end{align*}
\end{widetext}

As a result, to measure the original observables characterized by Eqs.\eqref{AliceEigenvector} and \eqref{BobEigenvector}, one only needs to measure the quantum system qubit by qubit, from $j=n$ to $j=1$, which can obtain the original measurement outcome bit by bit. This implies that it is realistic to implement our protocol physically.

\section{The equivalence checking of multiple quantum circuits}

Now let us go one step further. Suppose we have $k\geq3$ quantum circuits $C_1$, $C_2$, ..., $C_k$, and again we want to know whether they are equivalent to each other. Apparently, we can solve the problem by comparing these quantum circuits pair by pair. But if we are unlucky, we need to run the above two-circuit protocol for $k-1$ times. With the success in two-circuit case, we may wonder, can we design a similar protocol such that a proper $k$-partite Bell inequality allows us to solve the multi-circuit problem in one go? We show that, at least for the case that $k$ is odd, this is impossible.

Recall that a key part of our protocol is find a $k$-partite quantum state $\ket{\psi_k}$ and a certain Bell inequality such that $\ket{\psi_k}$ violates it maximally. Furthermore, $\ket{\psi_k}$ has to satisfy that condition that for any local unitary matrix $U$, it holds that $(U\otimes U\otimes \cdots \otimes U)\ket{\psi_k}=\ket{\psi_k}$.

For simplicity, we now suppose that for each party the local dimension is $2$, and the following argument is easy to be generalized to high-dimensional cases. Then we have that
\begin{equation*}
(\sigma_x\otimes \sigma_x\otimes \cdots \otimes \sigma_x)\ket{\psi_k}=\ket{\psi_k}
\end{equation*}
and
\begin{equation*}
(\sigma_z\otimes \sigma_z\otimes \cdots \otimes \sigma_z)\ket{\psi_k}=\ket{\psi_k},
\end{equation*}
where $\sigma_x$ and $\sigma_z$ are Pauli matrices. However, since $k$ is odd, $\sigma_x\otimes \sigma_x\otimes \cdots \otimes \sigma_x$ and $\sigma_z\otimes \sigma_z\otimes \cdots \otimes \sigma_z$ anticommute, which means that $\ket{\psi_k}$ is the zero vector, a contradiction.

Therefore, when $k$ is odd, we cannot generalize our two-circuit protocol to solve the equivalence checking problem in one go. However, we cannot rule out this possibility for the case that $k$ is even, where the major challenge is to find a desirable multipartite Bell inequality. We leave this for future work.

\section{Discussion}

In this paper, we have proposed a protocol for black-box equivalence checking of quantum circuits, where the key quantum property we have utilized is quantum nonlocality. We have proved the correctness of our protocol analytically and numerically. Particularly, we have shown that for any given strength of observed quantum nonlocality, the distance between two compared quantum circuits can be estimated accurately in an analytical manner. Furthermore, it turns out that the computational cost of our protocol is independent in the size of compared quantum circuits. Our work can be regarded as a nontrivial application of quantum nonlocality in the area of quantum engineering, and we hope this protocol can be applied in future quantum industries.

\bibliography{apssamp}

\begin{acknowledgments}
This work was supported by the National Key R\&D Program of China, Grants No. 2018YFA0306703, 2021YFE0113100, and the National Natural Science Foundation of China, Grant No. 61832015.
\end{acknowledgments}

\section*{Appendix A: The proof for Lemma 1}

\begin{lemma*}
Suppose $\ket{\psi}$ is a $d\times d$ quantum state orthogonal to $\ket{\Phi_d}$. Then
\begin{equation}
-m\leq I_{d,m}(\ket{\psi})\leq m(d-2).
\end{equation}
\end{lemma*}
\begin{proof}
    Recall that $A_x^1=\sum_{a=0}^{d-1}\omega^a|a\rangle_{xx}\langle a|$, and $\bar{B}_x^1=A_x^*$, where
    $$
    |a\rangle_x=\frac{1}{\sqrt{d}}\sum_{k=0}^{d-1}\exp\left[\frac{2\pi \mathbf i}{d}k(a-\alpha_x)\right]|k\rangle.
    $$
    Then it holds that
    \begin{equation*}
        \begin{split}
            A_x^1\otimes \bar{B}_x^1=\sum_{a=0}^{d-1}\sum_{b=0}^{d-1}\omega^{a-b}|a\rangle_{xx}\langle a|\otimes|(b\rangle_{xx}\langle b|)^*,
        \end{split}
    \end{equation*}
    and
    $$
    |b\rangle^*_x=\frac{1}{\sqrt{d}}\sum_{k=0}^{d-1}\exp\left[-\frac{2\pi \mathbf i}{d}k(b-\alpha_x)\right]|k\rangle.
    $$
    For a fixed $x$, let $|\psi\rangle=\sum_{a=0}^{d-1}\sum_{b=0}^{d-1}\beta_{ab,x}|a\rangle_x|b\rangle_x^*$. Then we have that
    $$
    \langle\psi|A_x^1\otimes\bar{B}_x^1|\psi\rangle=\sum_{a=0}^{d-1}\sum_{b=0}^{d-1}|\beta_{ab,x}|^2\omega^{a-b},
    $$
    and
    \begin{equation*}
        \begin{split}
            \sum_{l=1}^{d-1}\langle\psi|A_x^l\otimes\bar{B}_x^l|\psi\rangle&=\sum_{l=1}^{d-1}\sum_{a=0}^{d-1}\sum_{b=0}^{d-1}|\beta_{ab,x}|^2\omega^{l(a-b)}\\
            &=\sum_{a=0}^{d-1}|\beta_{aa,x}|^2(d-1)+\sum_{a\neq b}|\beta_{ab,x}|^2(-1)\\
            &=d\sum_{a=0}^{d-1}|\beta_{aa,x}|^2-1,
        \end{split}
    \end{equation*}
    which implies that
    $$
    \sum_{x=1}^m\sum_{l=1}^{d-1}\langle\psi|A_x^l\otimes\bar{B}_x^l|\psi\rangle=d\sum_{x=1}^m\sum_{a=0}^{d-1}|\beta_{aa,x}|^2-m.
    $$
    Then it is not hard to see that
    $$
    I_{d,m}(\ket{\psi})\ge -m.
    $$

    At the same time, we let $|\psi\rangle=\sum_{k=0}^{d-1}\sum_{j=0}^{d-1}\gamma_{kj}|k\rangle|j\rangle$. As it is orthogonal to $\ket{\Phi_d}$, we obtain that
    $$
    \sum_{k=0}^{d-1}\gamma_{kk}=0.
    $$
    Note that
    $$
    |a\rangle_x|a\rangle_x^*=\frac{1}{d}\sum_{k=0}^{d-1}\sum_{j=0}^{d-1}\exp\left[\frac{2\pi \mathbf i}{d}(k-j)(a-\alpha_x)\right]|k\rangle|j\rangle,
    $$
    thus we have
    $$
    \beta_{aa,x}=\frac{1}{d}\sum_{k=0}^{d-1}\sum_{j=0}^{d-1}\exp\left[\frac{2\pi \mathbf i}{d}(j-k)(a-\alpha_x)\right]\gamma_{kj}.
    $$

    As a result,
    \begin{widetext}
        \begin{equation*}
            \begin{split}
                &\sum_{x=1}^m\sum_{l=1}^{d-1}\langle\psi|A_x^l\otimes\bar{B}_x^l|\psi\rangle\\
                =&d\sum_{x=1}^m\sum_{a=0}^{d-1}|\beta_{aa,x}|^2-m\\
                =&\frac{1}{d}\sum_{x=1}^m\sum_{a=0}^{d-1}\left|\sum_{k=0}^{d-1}\sum_{j=0}^{d-1}\exp\left[\frac{2\pi \mathbf i}{d}(j-k)(a-\alpha_x)\right]\gamma_{kj}\right|^2-m\\
                =&\frac{1}{d}\sum_{x=1}^m\sum_{a=0}^{d-1}\left|\sum_{r=0}^{d-1}\left(\sum_{k=0}^{d-r-1}\exp\left[\frac{2\pi \mathbf i}{d}r(a-\alpha_x)\right]\gamma_{k(k+r)}+\sum_{k=d-r}^{d-1}\exp\left[\frac{2\pi \mathbf i}{d}(r-d)(a-\alpha_x)\right]\gamma_{k(k+r-d)}\right)\right|^2-m\\
                =&\frac{1}{d}\sum_{x=1}^m\sum_{a=0}^{d-1}\left|\sum_{r=0}^{d-1}\exp\left[\frac{2\pi \mathbf i}{d}r(a-\alpha_x)\right]\left(\sum_{k=0}^{d-r-1}\gamma_{k(k+r)}+\sum_{k=d-r}^{d-1}\exp\left[2\pi \mathbf i\alpha_x\right]\gamma_{k(k+r-d)}\right)\right|^2-m\\
                =&\sum_{x=1}^m\left|V_x\vec{z}_x\right|^2-m,
            \end{split}
        \end{equation*}
    \end{widetext}
    where we have defined the matrix $V_x$ and the vector $\vec{z}_x$ by setting their entries to be
    \begin{equation*}
        \begin{split}
            (\vec{z}_x)_r&=\sum_{k=0}^{d-r-1}\gamma_{k(k+r)}+\sum_{k=d-r}^{d-1}\exp\left[2\pi \mathbf i\alpha_x\right]\gamma_{k(k+r-d)},\\
            (V_x)_{a,r}&=\frac{\exp\left[\frac{2\pi \mathbf i}{d}r(a-\alpha_x)\right]}{\sqrt{d}}.
        \end{split}
    \end{equation*}

    It can be verified that $V_x$ is unitary, then we have $|V_x\vec{z}_x|=|\vec{z}_x|$. So
    \begin{equation*}
        \begin{split}
            &\sum_{x=1}^m\left|V_x\vec{z}_x\right|^2-m\\
            =&\sum_{x=1}^m\left|\vec{z}_x\right|^2-m\\
            =&\sum_{x=1}^m\sum_{r=0}^{d-1}\left|\sum_{k=0}^{d-r-1}\gamma_{k(k+r)}+\sum_{k=d-r}^{d-1}\exp\left[2\pi \mathbf i\alpha_x\right]\gamma_{k(k+r-d)}\right|^2-m.
        \end{split}
    \end{equation*}

    It can be verified that $\forall a,b,\lambda_k\in\mathbb{C}$, if $\sum_{k}\lambda_k=0$ and $|\lambda_k|=1$, we have $\sum_{k=1}^m|a+\lambda_kb|^2=m(|a|^2+|b|^2)$. Then it holds that
    \begin{align*}
            &\sum_{x=1}^m\sum_{r=0}^{d-1}\left|\sum_{k=0}^{d-r-1}\gamma_{k(k+r)}+\sum_{k=d-r}^{d-1}\exp\left[2\pi \mathbf i\alpha_x\right]\gamma_{k(k+r-d)}\right|^2-m\\
            =&m\sum_{r=0}^{d-1}\left(\left|\sum_{k=0}^{d-r-1}\gamma_{k(k+r)}\right|^2+\left|\sum_{k=d-r}^{d-1}\gamma_{k(k+r-d)}\right|^2\right)-m\\
            =&m\sum_{r=1}^{d-1}\left(\left|\sum_{k=0}^{d-r-1}\gamma_{k(k+r)}\right|^2+\left|\sum_{k=d-r}^{d-1}\gamma_{k(k+r-d)}\right|^2\right)-m\\
            \leq&m\sum_{r=1}^{d-1}\left((d-r)\sum_{k=0}^{d-r-1}\left|\gamma_{k(k+r)}\right|^2+r\sum_{k=d-r}^{d-1}\left|\gamma_{k(k+r-d)}\right|^2\right)-m\\
            \leq&m(d-1)\sum_{r=1}^{d-1}\left(\sum_{k=0}^{d-r-1}\left|\gamma_{k(k+r)}\right|^2+\sum_{k=d-r}^{d-1}\left|\gamma_{k(k+r-d)}\right|^2\right)-m\\
            \le&m(d-2).
    \end{align*}

\end{proof}

\section*{Appendix B: The Bell expression value for a random pure state}

\begin{lemma*}
    Given $0<\delta<1$. Suppose $|\psi\rangle$ is a $d\times d$ quantum state, which as a unit vector is chosen uniformly at random on the $d^2$-dimensional real unit sphere. Then with the probability of no less than $1-\delta$ it holds that
    \begin{equation}
        I_{d,m}(|\psi\rangle)\le m\sqrt{\frac{4}{3d\delta}}.
    \end{equation}
\end{lemma*}

\begin{proof}
    In the proof of Lemma \ref{lem:eigenvalues}, we have already known that if we let $|\psi\rangle=\sum_{k=0}^{d-1}\sum_{j=0}^{d-1}\gamma_{kj}|k\rangle|j\rangle$, we have
    \begin{align*}
        I_{d,m}(|\psi\rangle)=m\sum_{r=0}^{d-1}\left(\left|\sum_{k=0}^{d-r-1}\gamma_{k(k+r)}\right|^2+\left|\sum_{k=d-r}^{d-1}\gamma_{k(k+r-d)}\right|^2\right)-m.
    \end{align*}

    Let
    \begin{align*}
        g_{r1} &= \left|\sum_{k=0}^{r-1}\gamma_{k(k+d-r)}\right|^2, 1\le r\le d-1\\
        g_{r2} &= \left|\sum_{k=d-r}^{d-1}\gamma_{k(k+r-d)}\right|^2, 1\le r\le d-1\\
        g_d &= \left|\sum_{k=0}^{d-1}\gamma_{kk}\right|^2\\
        g &= \sum_{r=1}^{d-1}g_{r1}+\sum_{r=1}^{d-1}g_{r2}+g_d.
    \end{align*}

    Then
    \begin{align*}
        I_{d,m}(|\psi\rangle)=mg-m.
    \end{align*}

    Now let us figure out the expectation and variance of $g$. Due to symmetry, $E(g_{r1})=E(g_{r2})$. Then
    \begin{align*}
        E(g)&=\sum_{r=1}^{d-1}E(g_{r1})+\sum_{r=1}^{d-1}E(g_{r2})+E(g_d)\\
        &=2\sum_{r=1}^{d-1}E(g_{r1})+E(g_d).
    \end{align*}

    Due to symmetry again, $E(\gamma_{ij}\gamma_{kl})=E(-\gamma_{ij}\gamma_{kl})=0$, when $ij\neq kl$. And $E(|\gamma_{ij}|^2)=1/d^2$. Thus we have
    \begin{align*}
        E(g_{r1})&=E\left(\left|\sum_{k=0}^{r-1}\gamma_{k(k+d-r)}\right|^2\right)\\
        &=\sum_{k=0}^{r-1}E\left(\left|\gamma_{k(k+d-r)}\right|^2\right)\\
        &=\frac{r}{d^2}.
    \end{align*}

    Then
    \begin{align*}
        E(g)&=2\sum_{r=1}^{d-1}E(g_{r1})+E(g_d)\\
        &=2\sum_{r=1}^{d-1}\frac{r}{d^2}+\frac{1}{d}\\
        &=1.
    \end{align*}

    Before figuring out the variance, we need some auxiliary expectations. Denote the unit sphere in $d^2$-dimensional real space by $D$. Denote its surface area by $S_D=2\pi^{\frac{n}{2}}/\Gamma(\frac{n}{2})$. Then

    \begin{align*}
        E(|\gamma_{ij}|^4)&=\int_D |\gamma_{ij}|^4/S_D \mathrm{d}S\\
        &=\frac{1}{S_D}\int_0^{2\pi}\int_0^\pi\cdots\int_0^\pi \cos^4(\theta_1)\sin^{d^2-2}(\theta_1)\sin^{d^2-3}(\theta_2)\dots\sin(\theta_{d^2-2})\mathrm{d}\theta_1\dots \mathrm{d}\theta_{d^2-2}\mathrm{d}\theta_{d^2-1}\\
        &=\frac{3}{d^4+2d^2}.
    \end{align*}

    Similarly, we have
    \begin{align*}
        E(|\gamma_{ij}|^2|\gamma_{kl}|^2)=\frac{1}{d^4+2d^2},\text{ if }ij\neq kl.
    \end{align*}

    Note that due to symmetry, all the expectation containing odd power of $\gamma_{ij}$ is 0. Thus
    \begin{align*}
        E(g_{r1}^2)&= E\left(\left|\sum_{k=0}^{r-1}\gamma_{k(k+d-r)}\right|^4\right)\\
        &= \sum_{k=0}^{r-1}E\left(\left|\gamma_{k(k+d-r)}\right|^4\right)+3\sum_{k_1\neq k_2}E\left(\left|\gamma_{k_1(k_1+d-r)}\right|^2\left|\gamma_{k_2(k_2+d-r)}\right|^2\right)\\
        &=\frac{3r}{d^4+2d^2}+\frac{3r(r-1)}{d^4+2d^2}\\
        &=\frac{3r^2}{d^4+2d^2}.
    \end{align*}

    Similarly, we have

    \begin{align*}
        E(g_{r_11}g_{r_21})=\frac{r_1r_2}{d^4+2d^2}.
    \end{align*}

    Then
    \begin{align*}
        E(g^2)&=E\left(\left(\sum_{r=1}^{d-1}g_{r1}+\sum_{r=1}^{d-1}g_{r2}+g_d\right)^2\right)\\
        &=\frac{\left(\sum_{r=1}^{d-1}r+\sum_{r=1}^{d-1}r+d\right)^2}{d^4+2d^2}+2\frac{\sum_{r=1}^{d-1}r^2+\sum_{r=1}^{d-1}r^2+d^2}{d^4+2d^2}\\
        &=1+2\frac{(d-1)d(2d-1)}{3(d^4+2d^2)}\\
        &\le1+\frac{4}{3d}.
    \end{align*}

    Thus we conclude that
    \begin{align*}
        Var(g)=E(g^2)-E(g)^2\le\frac{4}{3d}.
    \end{align*}

    Then by chebyshev's inequality, with the probability of no less than $1-\delta$, we have $g\le1+\sqrt{\frac{4}{3d\delta}}$. That is
    \begin{align*}
        I_{d,m}(|\psi\rangle)=&mg-m\\
        \le&m\sqrt{\frac{4}{3d\delta}}.
    \end{align*}

\end{proof}

\end{document}